\documentclass[12pt]{article}

\usepackage[colorlinks=false]{hyperref}
\usepackage[dvips]{graphicx}
\usepackage[usenames,dvipsnames]{color}
\usepackage{epsfig}
\usepackage{rotating}
\usepackage{amssymb}
\usepackage{amsmath,amsfonts}
\usepackage[latin1]{inputenc}
\usepackage{verbatim}
\usepackage[tableposition=top,font=small,labelfont=bf,format=hang]{caption}
\usepackage{booktabs}
\usepackage{multirow}
\usepackage{geometry}
\newtheorem{defi}{Definition}
\newtheorem{lem}{Lemma}
\newtheorem{thm}{Theorem}
\newtheorem{prop}{Proposition}

\newtheorem{rem}{Remark}
\newtheorem{coro}{Corollary}
\newtheorem{ex}{Example}

\newcommand{\Z}{\mathbb{Z}}
\newcommand{\Q}{\mathbb{Q}}

\newenvironment{proof}{\noindent\textbf{Proof.}\quad}
{\hspace{\stretch{1}}%
\rule{1ex}{1ex}\\}

\begin{document}

\title{The covering radius of Butson Hadamard codes for the homogeneous metric}

\author{Xingxing Xu, Minjia Shi\thanks{Xingxing Xu and Minjia Shi are with the Key Laboratory of Intelligent Computing Signal Processing, Ministry of Education, School of Mathematical Sciences, Anhui
		University, Hefei 230601, China. Email: xuxingxing03@163.com, smjwcl.good@163.com}, Patrick Sol\'e\thanks{Institut de Math\'ematiques de Marseille (CNRS, Aix-Marseille University), Marseilles, France. Email: sole@enst.fr}}
		\date{}
	\maketitle	
	
		\begin{abstract}
		Butson matrices are complex Hadamard matrices with entries in the complex roots of unity of given order. There is an interesting code in phase space related to this matrix (Armario et al. 2023).
		 We study the covering radius of Butson Hadamard codes for the homogeneous metric, a metric defined uniquely, up to scaling, for a commutative ring alphabet that is Quasi Frobenius.
		 An upper bound is derived by an orthogonal array argument. A lower bound relies on the existence of bent sequences in the sense of (Shi et al. 2022). This latter bound generalizes a bound of (Armario et al. 2025) for the Hamming metric.
		\end{abstract}

\section{Introduction}
Let \( H \) be an \( n \times n \) matrix with complex entries of modulus 1. Hadamard's celebrated theorem on maximal determinants asserts that \( \det(HH^*) \le n^n \), we see in \cite{BEHC} for a recent survey of relevant results. Equality holds exactly when the rows of \( H \) are pairwise orthogonal under the Hermitian inner product, and such a matrix \( H \) is called a {\bf Hadamard matrix}. Let \( \zeta_q = \exp(2\pi\sqrt{-1}/q) \) denote the complex \( q \)-th root of unity, and let \( \langle \zeta_q \rangle = \{\zeta_q^j\}_{0 \le j \le q-1} \) represent the full set of \( q \)-th roots of unity. A Hadamard matrix with entries in \( \langle \zeta_q \rangle \) is called a {\bf Butson Hadamard matrix} of order \( n \) and phase \( q \). We use \( \mathrm{BH}(n,q) \) to denote the family of all such matrices. It is worth noting that the terms Hadamard matrix and complex Hadamard matrix are frequently used to refer to elements of  \( \mathrm{BH}(n,2) \) and \( \mathrm{BH}(n,4) \) respectively. With  a matrix  in \( \mathrm{BH}(n,q) \) can be attached a code of length $n$ over $\Z_q$ called the {\bf Butson Hadamard} code \cite{ABE}. In this paper we study the covering performance of these codes based on two properties: the existence of bent vectors for the matrix and the strength of the code as an orthogonal array. Like in the case of Sylvester type Hadamard matrices \cite{M}, bent vectors provide a lower bound, and moment arguments based on the strength yield upper bounds in the spirit of the Norse bounds \cite{MH}.

The initial study of bent vectors for Sylvester type Hadamard matrices was motivated by applications in cryptography \cite{M}, and coding theory \cite{MH}. A generalization of the definition of bent vector to all Hadamard matrices was studied for real Hadamard matrices and for
BH(n,4) matrices by Sol\'e, Cheng and coauthors \cite{SLC,SCGR} under the name of {\bf bent sequence}
attached to a Hadamard matrix. It has been shown computationally that there exist
Butson Hadamard matrices which admit conjugate self-dual bent vectors but do not
admit self-dual bent vectors \cite{SL}. The covering property of Butson Hadamard codes which admit bent vectors has been investigated in \cite{AEKC}.

As one of the most crucial metrics for codes, the {\bf homogeneous metric} has been investigated over various rings \cite{SAS}. Initially it was introduced in the context of residue class rings \cite{CH}, it can be more generally defined for quasi-Frobenius rings \cite{NK}. For any given alphabet ring, this metric is uniquely determined (up to scaling) by the lattice of ideals inherent to that ring. In the specific case of the ring $\mathbb{Z}_4$, the homogeneous weight coincides with the Lee weight and corresponds to the Hamming weight of the associated Gray map as established in \cite{matrix}. An important special case of an alphabet is the ring $\Z_{{2}^k}$ where $k$ is an integer satisfying $k\ge 2$. Within this framework, the homogeneous weight of a vector over $\Z_{{2}^k}$ is equivalent to the Hamming weight of its image under Carlet's Gray map \cite{C}. Investigations into perfect codes over $\Z_{{2}^k}$ were initiated in \cite{OS} for $k=2$ and subsequently extended to the general case of arbitrary $k$ in \cite{SOS}. The paper has given a formula as an exponential sum for a homogeneous weight \cite{VW}, which can be used to obtain the definition of homogeneous weight of the element in $\mathbb{Z}_q$ for any positive integer.

The material is arranged as follows. Section 2 collects the definitions and notations needed for the rest of the paper. Section 3 studies the lower bound on the covering radius of Butson Hadamard codes which admits a self-dual bent vector. Section 4 provides upper bounds for self-complementary BH-codes which satisfy strength 2, often called the Norse bounds. Section 5 concludes the article.
\section{Preliminaries and Notations}
Let $\Z_q$ be the ring of integers
 modulo $q$ with $q>1$, and denote by $\Z_q^n$ the set of $n$-tuples over $\Z_q$. We use
 bold notation $\mathbf{x}=(x_1,\dots,x_n)\in\Z_q^n$ to denote vectors (or codewords) in $\Z_q^n$. For any complex number $z\in \mathbb{C}$, the complex conjugate is denoted by $\bar{z}$. We denote the set of $n \times n$ matrices with entries in a set $X$ by $\mathcal{M}_n(X)$. The transpose of a matrix $M$ is denoted $M^{\mathrm{T}}$,
 the conjugate transpose by $M^*$.
 \subsection{Butson Hadamard matrices and Bent Vectors}
 A {\em Butson Hadamard} matrix of order $n$ and phase $q$ is a matrix $H$ with entries  in $\langle\zeta_q\rangle$ such that $HH^*=nI_n$, where $I_n$ denotes the identity matrix of order $n$. We write
 $\mathrm{BH}(n,q)$ for the set of such matrices. The simplest examples of Butson matrices are the Fourier matrices $$F_n=[\zeta_q^{(i-1)(j-1)}]_{1\le i,j\le n}\in \mathrm{BH}(n,n).$$ Hadamard
 matrices of order $n$, as they are usually defined, are the elements of $\mathrm{BH}(n,2)$.
 The phase and orthogonality of a $\mathrm{BH}(n,q)$ is preserved by multiplication on
 the left or right by a $n\times n$ monomial matrix with non-zero entries in the $q$-th
 roots of unity. A pair of monomial matrices $(P,Q)$ with non-zero 
 entries in $\langle\zeta_q\rangle$ acts on the set $\mathrm{BH}(n,q)$ by $H(P,Q)=PHQ^*$.  This action preserves the set $\mathrm{BH}(n,q)$, and matrices in the same orbit of the action are said to be Hadamard-equivalent. 
 
 A Butson matrix $H \in \mathrm{BH}(n,q)$ is conveniently represented in logarithmic form, that is, the matrix $H=[\zeta _q^{{\varphi _{i,j}}}]_{i,j = 1}^n$ is represented by the matrix $L(H)=[\varphi _{i,j} \bmod q]_{i,j = 1}^n$ with the convention that $L_{i,j}\in \Z_q$ for all $i,j\in\{1,\dots,n\}$.
 \begin{ex}
 	The following is a matrix $H\in \mathrm{BH}(4,8)$, displayed in logarithmic form
 	$$L(H)=\left( {\begin{array}{*{20}{c}}
 			0&0&0&0 \\ 
 			0&2&4&6 \\ 
 			0&4&0&4 \\ 
 			0&6&4&2 
 	\end{array}} \right).$$
 \end{ex}
 
 Observe that the matrix above is in dephased form, that is, its first row and column are all 0. Every matrix $H\in \mathrm{BH}(n,q)$ is Hadamard-equivalent to a dephased matrix. Throughout this paper all matrices are assumed to be dephased, with the exception of circulant matrices.
 
 The following part describes the definition of a bent vector attached to Butson Hadamard matrix $H\in\mathrm{BH}(n,q)$.
 \begin{defi}
 	Let $H\in \mathrm{BH}(n,q)$, and let $\mathbf{x}$ be a vector of length $n$ with entries in $\langle\zeta_q\rangle$. We say that $\mathbf{x}$ is H-bent if and only if $$H\mathbf{x}=\sqrt{n}\mathbf{y},$$ where all entries of $\mathbf{y}$ are complex numbers of norm 1. Furthermore, $\mathbf{x}$ is self-dual H-bent
 	if $\mathbf{y}=\lambda\mathbf{x}$ for $\lambda$ of modulus 1, and conjugate self-dual H-bent if $\mathbf{y}=\lambda\mathbf{x}$, where $\bar{\mathbf{x}}$ is the
 	vector of complex conjugates of elements of $\bar{\mathbf{x}}$.
 \end{defi}

\subsection{Butson Hadamard codes}
Let $H \in \mathrm{BH}(n,q)$. We denote by $R_H$ the $\Z_q$-code of length $n$ consisting
of the rows of $L(H)$, and we denote by $C_H$ the $\Z_q$-code defined as $C_H = \cup _{\alpha\in\Z_q}
(R_H+\alpha\mathbf{1})$, where $\mathbf{1}$ denotes the all-one vector.
The code $C_H$ over $\Z_q$ is called a Butson Hadamard code (briefly, $\mathrm{BH}$-code). Any subset of $\langle\zeta_q\rangle^n$ for some integer $n$ is called a {\it polyphase} code, following \cite{EZ}. 

A {\em linear} code $C$ of length $n$ over $\Z_q$ is an additive subgroup of $\Z_q^n$. An element of $C$ is called a {\em codeword} of $C$. The {\em homogeneous weight} of $x\in\Z_{p^k}$ for prime $p$ and $k\ge1$ \cite{HY} is defined by $$wt(x)=
\begin{cases}
	0 & \text{if $x=0,$} \\
	p^{k-1} & \text{if $x\in p^{k-1}\Z_{p^k}\backslash\{0\},$} \\
	(p-1)p^{k-2} & \text{if $x\in\Z_{p^k}\backslash p^{k-1}\Z_{p^k}.$}
\end{cases}$$The homogeneous weight $wt(\mathbf{x})$ of $\mathbf{x}$ is just the rational sum of the homogeneous weights of its coordinates. The homogeneous distance $d(\mathbf{x},\mathbf{y})$ between two vectors $\mathbf{x}$ and $\mathbf{y}$ is $wt(\mathbf{x}-\mathbf{y})$. The {\em minimum distance} for the homogeneous distance $d$ of $C$ is the smallest homogeneous weights among all nonzero codewords of $C$.

\begin{defi}
	Let $C_H$ be a Butson Hadamard code of length $n$ over $\Z_q$ coming from $H \in \mathrm{BH}(n,q)$. The covering radius of $C_H$ with homogeneous distance $d$ is defined as:$$r(C_H)=\max\{d(\mathbf{x},C_H)\mid\mathbf{x}\in\Z_q^n\}=\max\limits_{\mathbf{x}\in\Z_q^n}\min\limits_{\mathbf{y}\in C_H}d(\mathbf{x},\mathbf{y}).$$
\end{defi}

\section{The lower bound for BH-codes}
We study lower bound on the covering radius of Butson Hadamard codes $C_H$ for an arbitrary $H\in \mathrm{BH}(n,q)$, and under the hypothesis that $H$ admits a bent vetor. The following theorem derives the homogeneous weight of $x\in\Z_q$ for a positive integer $q$.
\begin{thm}
	For any $x$ in $\Z_q$ we have $$wt(x)=\lambda\bigg(1-\frac{1}{\phi(q)}\sum\limits_{a\in \Z_q^{\times}}\bigg(\prod\limits_{i=1}^{s}\zeta_{p_i}^{a_ix_i}\bigg)\bigg),$$ where $\Z_q^{\times}$ is the group of units of $\Z_q$ and $\phi$ is the Euler function. 
\end{thm}
\begin{proof}
	For any integer $q>1$ we have a standard decomposition $q=p_1^{e_1}\cdots p_s^{e_s}$, where $p_1,\dots,p_s$ are distinct primes and $e_i>0$ for $i=1,\dots,s.$ The classical number-theoretical Euler $\phi$-function is:$$\phi(q)=\prod\limits_{i=1}^{s}\big(p_i^{e_i}-p_i^{e_i-1}\big).$$  
	According to the Honold's formula for homogeneous weights in \cite{H}, for any $x\in \Z_q$ we have $$wt(x)=\lambda\bigg(1-\frac{1}{|\Z_q^{\times}|}\sum\limits_{a\in\Z_q^{\times}}\chi(ax)\bigg),$$ where $\Z_q^{\times}$ is the group of units of $\Z_q$ and $\chi$ is a generating character of $\Z_q.$
	
	From the fundamental theorem of finite abelian groups, $\Z_q$ decomposes uniquely into a finite direct sum of subgroups of prime power order, namely, $$\Z_q=\Z_{p_1^{e_1}}\times\cdots\times\Z_{p_s^{e_s}}.$$ For any $a=(a_1,\dots,a_s)\in\Z_q^{\times},x=(x_1,\dots,x_s)\in\Z_q,$ the generating character $\chi$ has the form $$\chi(ax)=\prod\limits_{i=1}^{s}\zeta_{p_i}^{a_ix_i},$$ where $\zeta_{p_i}$ is a $p_i$-th root of unity. Then the result follows from $|\Z_q^{\times}|=\phi(q).$
\end{proof}
\begin{lem}\label{lemma:2}
	Let $\mathbf{v},\mathbf{w}\in {\langle {\zeta _q}\rangle ^n}$, then $$d(L(\mathbf{v}),L(\mathbf{w}))=n\phi(q)-T(\langle {\mathbf{v},\mathbf{w}}\rangle),$$ where $\langle\mathbf{v},\mathbf{w}\rangle$ denotes the Hermitian product between $\mathbf{v}$ and $\mathbf{w}$.
\end{lem}
\begin{proof}
	According to the previous theorem, for any $x\in\Z_q$ and $q=p_1^{e_1}\cdots p_s^{e_s},$ $$wt(x)=\phi(q)-\sum\limits_{a\in\Z_q^{\times}}\bigg(\prod\limits_{i=1}^s\zeta_{p_i}^{a_ix_i}\bigg),$$ where $\Z_q^{\times}$ is the group of units of $\Z_q$.
	
	Consider $\Z_q=\Z_{p_1^{e_1}}\times\cdots\times\Z_{p_s^{e_s}},$ we define a trace map from $\Q(\zeta_q)$ down to $\Q$ as: for any $a=(a_1,\dots,a_s)\in\Z_q^{\times},x=(x_1,\dots,x_s)\in\langle\zeta_q\rangle,$ $$T(x):=\sum\limits_{a\in \Z_q^{\times}}\bigg(\prod\limits_{i=1}^sx_i^{a_i}\bigg).$$ Since $\mathbf{v},\mathbf{w}\in {\langle {\zeta _q}\rangle ^n},$ then $L(\mathbf{v}),L(\mathbf{w})\in\Z_q^n.$ We denote $L({v_i})-L({w_i})$ by $L_i$, and $L_i=(L_{i_1},\dots,L_{i_s})$ for $L_{i_j}\in\Z_{p_i^{e_i}},j=1,\dots,s.$
	\begin{align*}
		wt(L(\mathbf{v})-L(\mathbf{w})) 
		= & \sum\limits_{i=1}^{n} wt(L({v_i})-L({w_i})) \\
		= & \sum\limits_{i=1}^n\bigg(\phi(q)-\sum\limits_{a\in\Z_q^{\times}}\bigg(\prod\limits_{j=1}^s\zeta_{p_j}^{a_jL_{i_j}}\bigg)\bigg) \\
		= & n\phi(q)-\sum\limits_{i=1}^n\Bigg[\sum\limits_{a\in\Z_q^{\times}}\bigg(\prod\limits_{j=1}^s\Big(\zeta_{p_j}^{L_{i_j}}\Big)^{a_j}\bigg)\Bigg] \\
		= & n\phi(q)-\sum\limits_{i=1}^nT\big(\zeta_{q}^{(L({v_i})-L({w_i}))}\big) \quad\text{(Since the linearity of $T$)}\\
		= & n\phi(q)-T(\langle {\mathbf{v},\mathbf{w}}\rangle).
	\end{align*}
	Then the proof is complete.
\end{proof}
\begin{rem}
	To alleviate calculations, we use the homogeneous weight by taking $\lambda=\phi(q)$ in the proof of previous lemma. In fact, when $q=p^k$ for prime $p$ and $k\ge2,$ we take $\lambda=p^{k-1}-p^{k-2}$ to obtain the weight which satisfies the definition of the weight in reference \cite{HY}.
\end{rem}
\begin{thm}\label{thm:3}
	Suppose that $H \in \mathrm{BH}(n,q)$ admits a bent vector. Then$$r(C_H)\ge \phi(q)(n-\sqrt{n}).$$ Especially, when $q=p^k$ for prime $p$ and $k\ge2,$ then $$r(C_H)\ge \big(p^{k-1}-p^{k-2}\big)(n-\sqrt{n}).$$ 
\end{thm}
\begin{proof}
	Let $\mathbf{v}$ be a bent vector for $H$ and $\mathbf{w}$ be any row vector of $H$. By hypothesis, $\langle\mathbf{v},\mathbf{w}\rangle=\sqrt{n}\zeta_{q}^i$ for $i\in\{0,1,\dots,q-1\}.$ We notice that for $i=(i_1,\dots,i_s)\in\Z_q,a=(a_1,\dots,a_s)\in\Z_q^{\times},\zeta_q^i=\big(\zeta_{p_1}^{i_1},\dots,\zeta_{p_s}^{i_s}\big)\in\langle\zeta_q\rangle$, $$T(\langle\mathbf{v},\mathbf{w}\rangle)=\sqrt{n}T(\zeta_{q}^i)=\sqrt{n}\sum\limits_{a\in \Z_q^{\times}}\bigg(\prod\limits_{j=1}^s\big(\zeta_{p_j}^{i_j}\big)^{a_j}\bigg)=\sqrt{n}\sum\limits_{(j,q)=1}\zeta_q^{ij}$$$$\Rightarrow|T(\langle\mathbf{v},\mathbf{w}\rangle)|\le\sqrt{n}\phi(q)$$ where $\phi$ is the Euler function. By Lemma \ref{lemma:2}, we conclude that $$d(L(\mathbf{v}),L(\mathbf{w}))\ge n\phi(q)-T(\langle {\mathbf{v},\mathbf{w}}\rangle)=\phi(q)(n-\sqrt{n})$$ for any codeword $L(\mathbf{w})$. Thus $L(\mathbf{v})$ is a vector at distance at least  from all codewords, which concludes the proof.
\end{proof}

\begin{ex}
	We consider the Butson Hadamard code generated by the Fourier matrix $F_4\in\mathrm{BH}(4,4)$, displayed in logarithmic form: $$L(H)=\left( {\begin{array}{*{20}{c}}
			0&0&0&0 \\ 
			0&1&2&3 \\ 
			0&2&0&2 \\ 
			0&3&2&1 
	\end{array}} \right),$$ then $C_H$ is the following
set of 16 codewords: $$C_H=\{[i,i,i,i],[i,1+i,2+i,3+i],[i,2+i,i,2+i],[i,3+i,2+i,1+i]\mid i\in\Z_4\}.$$ By a Magma computation \cite{magma}, the covering radius of $C_H$ is exactly  $2$, which equals the lower bound in Theorem \ref{thm:3}.
\end{ex}

\section{The upper bound for BH-codes}
Here we extend the Norse bounds on the covering radius for the Hamming distance of binary codes \cite{MH} to BH-codes for the homogeneous distance over $\Z_q.$ To do this we need to recall what the strength of a code is.
\begin{defi}
	A code $C$ over $\Z_q$ has strength $s$ if each $s$-subset of the coordinates of the code contains all $s$-tuples a constant number of times.
\end{defi}

We proceed to derive bounds of precision increasing with the strength. We begin with strength one.
\begin{lem} \label{lemma:3}
	Let $C_H$ be a BH-code of length $n$ and let $\mathbf{v} \in \Z_q^n$, where $H\in\mathrm{BH}(n,q)$. If $C_H$ has strength $1$, then $$\sum_{\mathbf{u}\in \mathbf{v} + C_H}wt(\mathbf{u})=n\phi(q)|C_H|.$$ Especially, when $q=p^k$ for prime $p$ and $k\ge 2,$ then $$\sum_{\mathbf{u}\in \mathbf{v} + C_H}wt(\mathbf{u})=n\big(p^{k-1}-p^{k-2}\big)|C_H|.$$
\end{lem}
\begin{proof}
	Since $|{\mathbf{v} + C_H}| =|C_H|$ and $\mathbf{v}+C_H$ has strength $s$ if and only if $C_H$ has strength $s,$ it is enough to prove the lemma for $\mathbf{v}=\mathbf{0}.$
	
	Since $C_H$ has strength 1, we notice that each position has an equal number of elements in $\Z_q=\{0,1,\dots ,q-1\},$ and so each position contributes $|C_H|/q$ to the sum. Then from the equation \begin{align*}
		\sum\limits_{x\in\Z_q}wt(x)&=\sum\limits_{x\in\Z_q}\Bigg[\phi(q)-\sum\limits_{a\in\Z_q^{\times}}\bigg(\prod\limits_{i=1}^s\zeta_{p_i}^{a_ix_i}\bigg)\Bigg] =q\phi(q)-\sum\limits_{x\in\Z_q}\Bigg(\sum\limits_{a\in\Z_q^{\times}}\bigg(\prod\limits_{i=1}^{s}\zeta_{p_i}^{a_ix_i}\bigg)\Bigg) \\
		&=q\phi(q)-\sum\limits_{x\in\Z_q}\sum\limits_{a\in\Z_q^{\times}}\zeta_q^{ax}=q\phi(q)-\bigg(\sum\limits_{a\in\Z_q^{\times}}\zeta_q^a\bigg)\bigg(\sum\limits_{x\in\Z_q}\zeta_q^x\bigg)=q\phi(q),
	\end{align*}the result follows.
\end{proof}
\begin{thm}
	Let $C_H$ be a BH-code of length $n$ over $\Z_q$ and strength $1$, then $$r(C_H)\le n\phi(q).$$ In particular, when $q=p^k$ for prime $p$ and $k\ge 2,$ then $$r(C_H)\le n\big(p^{k-1}-p^{k-2}\big)$$
\end{thm}
\begin{proof}
	Let $\mathbf{v}\in \Z_q^n.$ By Lemma \ref{lemma:3}, the average distance between $\mathbf{v}$ and a codeword is $n\phi(q).$ Hence at least one codeword is within distance $n\phi(q)$ of $\mathbf{v}.$
\end{proof}

We know that a $\mathrm{BH}$-code generated by the Fourier matrix satisfies the strength of 2 with the minimum length when $n=p$. Therefore, we can construct a longer Butson Hadamard matrix based on the Fourier matrix by tensor product, and the generated $\mathrm{BH}$-code still satisfies the strength of 2. The following construction method is given:
\begin{prop}\label{prop:1}
	Let $F_p=[\zeta_p^{(i-1)(j-1)}]_{1\le i,j\le p}\in\mathrm{BH}(p,p)$, then $$F_p\otimes F_p=\left( {\begin{array}{*{20}{c}}
			{{F_p}}& \cdots &{{F_p}} \\ 
			\vdots & \ddots & \vdots  \\ 
			{{F_p}}& \cdots &{\zeta_p{F_p}} 
	\end{array}} \right)\in\mathrm{BH}(p^2,p)$$
	still satisfies the strength of $2$, where $\otimes$ represents the Kronecker product.
\end{prop}
\begin{proof}
	Let $H\triangleq F_p\otimes F_p$, since $$HH^*=(F_p\otimes F_p)(F_p\otimes F_p)^*=(F_p\otimes F_p)(F_p^*\otimes F_p^*)=(F_pF_p^*)\otimes(F_pF_p^*)$$ and $F_pF_p^*=pI_p$. Then $$HH^*=(pI_p)\otimes(pI_p)=p^2(I_p\otimes I_p)=p^2I_{p^2}.$$
	Combining the properties of Kronecker product, we infer that $H\in\mathrm{BH}(p^2,p)$ satisfies the strength of $2$.
\end{proof}
\begin{ex}
	Let $F_3=[\zeta_3^{(i-1)(j-1)}]_{1\le i,j\le 3} \in \mathrm{BH}(3,3)$, that is $$F_3=\left({\begin{array}{*{20}{c}}
			1&1&1 \\
			1&\zeta_3 &{{\zeta_3 ^2}} \\
			1&\zeta_3^2 &{{\zeta_3}}
	\end{array}}\right),L(F_3)=\left({\begin{array}{*{20}{c}}
			0&0&0 \\
			0&1&2 \\
			0&2&1
	\end{array}}\right).$$ We compute the Kronecker product $$F_3\otimes F_3=\left( {\begin{array}{*{20}{c}}
			1&1&1&1&1&1&1&1&1 \\ 
			1&\zeta &{{\zeta ^2}}&1&\zeta &{{\zeta ^2}}&1&\zeta &{{\zeta ^2}} \\ 
			1&{{\zeta ^2}}&\zeta &1&{{\zeta ^2}}&\zeta &1&{{\zeta ^2}}&\zeta  \\ 
			1&1&1&\zeta &\zeta &\zeta &{{\zeta ^2}}&{{\zeta ^2}}&{{\zeta ^2}} \\ 
			1&\zeta &{{\zeta ^2}}&\zeta &{{\zeta ^2}}&1&{{\zeta ^2}}&1&\zeta  \\ 
			1&{{\zeta ^2}}&\zeta &\zeta &1&{{\zeta ^2}}&{{\zeta ^2}}&\zeta &1 \\ 
			1&1&1&{{\zeta ^2}}&{{\zeta ^2}}&{{\zeta ^2}}&\zeta &\zeta &\zeta  \\ 
			1&\zeta &{{\zeta ^2}}&{{\zeta ^2}}&1&\zeta &\zeta &{{\zeta ^2}}&1 \\ 
			1&{{\zeta ^2}}&\zeta &{{\zeta ^2}}&\zeta &1&\zeta &1&{{\zeta ^2}} 
	\end{array}} \right),$$ represented in logarithmic form $$L(H)\triangleq L(F_3\otimes F_3)=\left( {\begin{array}{*{20}{c}}
			0&0&0&0&0&0&0&0&0 \\ 
			0&1&2&0&1&2&0&1&2 \\ 
			0&2&1&0&2&1&0&2&1 \\ 
			0&0&0&1&1&1&2&2&2 \\ 
			0&1&2&1&2&0&2&0&1 \\ 
			0&2&1&1&0&2&2&1&0 \\ 
			0&0&0&2&2&2&1&1&1 \\ 
			0&1&2&2&0&1&1&2&0 \\ 
			0&2&1&2&1&0&1&0&2 
	\end{array}} \right).$$ Then the BH-code $C$ generated by $L(H)$ over $\mathbb{Z}_3$ satisfies the strength $2$.
\end{ex}
\begin{lem}
	Let $C_H$ be a BH-code of length $n$ and let $\mathbf{v} \in \Z_p^n$, where $H\in\mathrm{BH}(n,p)$. If $C_H$ has strength $2$, then $$\sum_{\mathbf{u}\in \mathbf{v} + C_H}wt(\mathbf{u})^2=n^2(p-1)\big((p-1)|C_H|+p\big).$$
\end{lem}
\begin{proof}
	We only prove the lemma for $\mathbf{v}=\mathbf{0}.$ Since $C_H$ has strength 2, We know that $C_H$ must be constructed from the above proposition. Due to the particularity of the definition of homogeneous weight, we notice that in $C_H$, the weights of all other codewords are equal, that is $\dfrac{n}{p}\sum\limits_{x\in\Z_p}wt(x),$ except for the weights of codewords:$\{\overline{0},\overline{1},\overline{2},\dots,\overline{q-1}\}$, where $\overline{i}$ represent all $i$ codewords of length $n$. Then 
	\begin{align*}
		\sum\limits_{\mathbf{u} \in C_H} {wt{{( \mathbf{u} )}^2}}
		&= (|C_H|-p)\bigg(\dfrac{n}{p}\sum\limits_{x\in\Z_p}wt(x)\bigg)^2+\sum\limits_{x\in\Z_p}\big(n\cdot wt(x)\big)^2 \\
		&=n^2\phi(p)^2(|C_H|-p)+n^2\sum\limits_{x\in\Z_p}wt(x)^2 \\
		&=n^2\phi(p)^2(|C_H|-p)+n^2\sum\limits_{x\in\Z_p}\bigg(\phi(p)-\sum\limits_{a\in\Z_p^{\times}}\zeta_p^{ax}\bigg)^2 \\
		&=n^2\phi(p)^2|C_H|+n^2\sum\limits_{x\in\Z_p}\Bigg(\bigg(\sum\limits_{a\in\Z_p^{\times}}\zeta_p^{ax}\bigg)^2-2\phi(p)\bigg(\sum\limits_{a\in\Z_p^{\times}}\zeta_p^{ax}\bigg)\Bigg) \\
		&=n^2\phi(p)^2|C_H|+n^2\big(\phi(p)^2+\phi(p)\big)=n^2(p-1)\big((p-1)|C_H|+p\big).
	\end{align*}
The proof is completed.
\end{proof}
\begin{defi}
	A code $C$ over $\Z_q$ is self-complementary if $$\forall \mathbf{c}\in C, \forall \omega \in \Z_q, \bar{\mathbf{c}}^{\omega}=(\omega,\dots,\omega)+\mathbf{c}\in C.$$
\end{defi}

We observe that a code is self-complementary if it is closed under addition of $\alpha\mathbf{1}$ for all $\alpha\in\Z_q$ and the condition of strength 2 is clearly met by the construction of $C_H.$ The next theorem follows by a technique similar to that of Leducq \cite{RM}.
\begin{thm}\label{thm:1}
	The BH-code $C_H$ is a self-complementary code over $\Z_p$ of length $n$ and strength $2$, and $$r(C_H)\le n\Big(p-1-\Big\lceil {\sqrt{p/|C_H|}}\Big\rceil\Big).$$
\end{thm}
\begin{proof}
	Let $\mathbf{v}\in \Z_p^n$ such that its distance to any codeword is at least $r$, i.e. $\forall \mathbf{u} \in \mathbf{v} + C, wt(\mathbf{u}) \ge r$. Since $C_H$ is self-complementary, then $\bar{\mathbf{u}}^{\omega}\in \mathbf{v} + C_H$ and $wt(\bar{\mathbf{u}}^{\omega}) \ge r.$
	We have $$\sum_{\omega \in \Z_p}wt(\bar{\mathbf{u}}^{\omega})=n\bigg(\sum_{\omega \in \Z_p}wt(\omega)\bigg),$$ where
	$$\sum_{\omega \in \Z_p}wt(\omega)=p\phi(p).$$
	Assume $r>n(p-1)$, then \begin{align*}
		n(p-1)<r
		&\le wt( \mathbf{u} ) = n\bigg( {\sum\limits_{\omega  \in {\mathbb{Z}_p}} {wt( \omega  )} } \bigg) - \sum\limits_{\omega  \in \mathbb{Z}_p\backslash\{0\}} {wt( \bar{\mathbf{u}}^{\omega} )} \\
		&\le np(p-1) - (p-1)r < n(p-1).
	\end{align*}
	This leads to a contradiction. So we write $r=n(p-1)-\rho$ with $\rho \ge 0$ and we have
	\begin{align*}
		\sum_{\omega \in \Z_p}wt(\bar{\mathbf{u}}^{\omega})^2
		&=\sum_{\omega \in \Z_p\backslash\{0\}}wt(\bar{\mathbf{u}}^{\omega})^2+wt(\mathbf{u})^2 \\
		&=\sum_{\omega \in \Z_p\backslash\{0\}}wt(\bar{\mathbf{u}}^{\omega})^2+\bigg(n\bigg( {\sum\limits_{\omega  \in {\mathbb{Z}_p}} {wt( \omega  )} } \bigg) - \sum\limits_{\omega  \in \mathbb{Z}_p\backslash\{0\} } {wt( \bar{\mathbf{u}}^{\omega} )}\bigg)^2  \\
		&\le ( {p - 1} ){r^2} + {( {p - 1} )^2}{( {np - r} )^2} \\
		&=( {p - 1} ){(n(p-1)-\rho)^2} + {( {p - 1} )^2}{(n+\rho)^2} \\
		&=p\big(n^2(p-1)^2 + (p - 1){\rho ^2}\big). 
	\end{align*}
	Decomposing $C_H$ into cosets of the repetition code $\{(\omega,\dots, \omega) \mid \omega \in \Z_p\}$, we obtain $$\sum\limits_{\mathbf{u} \in \mathbf{v} + C_H} {wt{{( \mathbf{u} )}^2}} \le |C_H|\big(n^2(p-1)^2 + (p - 1){\rho ^2}\big).$$
	Then $$n^2(p-1)\big((p-1)|C_H|+p\big)\le |C_H|\big(n^2(p-1)^2 + (p - 1){\rho ^2}\big).$$ Hence $${\rho}^2\ge \frac{pn^2}{|C_H|}.$$ Since $\rho \ge 0,$
	$$r\le n\Big(p-1-\Big\lceil {\sqrt{p/|C_H|}}\Big\rceil\Big).$$ Then the proof is complete.
\end{proof}

\begin{ex}
	The Table \ref{table:4} computes the upper bounds in Theorem \ref{thm:1}, comparing them with the covering radius of $C_H$ constructing by the method in Proposition \ref{prop:1} for small $n$ and $p$. 
	\begin{table}[h!] 
		\centering
		\renewcommand{\arraystretch}{1.25}
		\caption{The covering radius $r(C_H)$ and the bounds of BH-codes over $Z_p$} \label{table:4}
		\begin{tabular}{c|c|c|c}
			\hline
			$k$ & $n$ & Covering radius & Upper bound \\
			\hline
			2 & 4 & 1 & 2 \\  \hline
			\multirow{2}{*}{3} & 3 & 2 & 4 \\ 
			& 9 & 10 & 15 \\  \hline
			\multirow{2}{*}{5} & 5 & 12 & 17 \\ 
			& 25 &  & 95 \\
			\hline
		\end{tabular}
	\end{table}
\end{ex}
\section{Conclusion and open problems}
In this paper, we have studied the homogeneous weight of an element in $\mathbb{Z}_q$ for any  integer $q>1$. Building on this study, we have derived a lower bound on the covering radius of BH-codes when the Butson Hadamard matrix admits a bent vector. We have also provided an example of a BH-code that meets the lower bound with equality.

In a second part, we have given a method for constructing BH-codes that are of strength of 2 over $\mathbb{Z}_p$. We have derived the analogue of the Norse bounds for this kind of Butson Hadamard codes in the context of the homogeneous metric. The main open problem of this research is to characterize
Butson matrices whose BH-codes have strength $2.$
\section*{Declarations}

\noindent\textbf{Data availability} All data produced is available from the authors upon request.  \\

\noindent\textbf{Conflict of Interest} The authors declare no conflict of interest.

\end{document}